\documentclass[11pt,a4paper]{article}

\usepackage[T1]{fontenc}
\usepackage[utf8]{inputenc}
\usepackage{palatino}
\usepackage[sc]{mathpazo}

\usepackage[margin=2.4cm,top=2.6cm,bottom=3cm]{geometry}

\usepackage{amsmath,amssymb,amsthm}
\usepackage{mathtools}
\usepackage{stmaryrd}

\usepackage{algorithm}
\usepackage{algpseudocode}
\algnewcommand{\LineComment}[1]{\State\(\triangleright\)\;\textit{#1}}

\usepackage{booktabs}
\usepackage{array}
\usepackage{multirow}
\usepackage{tabularx}

\usepackage{xcolor}
\usepackage{tikz}
\usepackage{tikz-3dplot}
\usepackage{pgfplots}
\pgfplotsset{compat=1.18}
\usetikzlibrary{
  arrows.meta, positioning, shapes.geometric,
  decorations.pathreplacing, decorations.markings,
  calc, backgrounds, patterns, fit, matrix,
  3d, perspective, shadings, fadings
}

\usepackage[most]{tcolorbox}
\tcbuselibrary{theorems,skins,breakable}

\usepackage{natbib}
\usepackage{hyperref}
\hypersetup{
  colorlinks=true,
  linkcolor=Blue!70!black,
  citecolor=ForestGreen!80!black,
  urlcolor=Teal
}
\usepackage{cleveref}

\usepackage{enumitem}
\usepackage{parskip}
\usepackage{caption}
\definecolor{Blue}{RGB}{31,119,180}
\definecolor{OrangeRed}{RGB}{214,39,40}
\definecolor{ForestGreen}{RGB}{44,160,44}
\definecolor{Purple}{RGB}{148,103,189}
\definecolor{Teal}{RGB}{23,190,207}
\definecolor{Sand}{RGB}{255,248,220}
\definecolor{BoxBlue}{RGB}{232,242,254}
\definecolor{BoxGreen}{RGB}{232,250,232}
\definecolor{DarkGray}{RGB}{50,50,50}

\tcbset{
  defstyle/.style={
    colback=BoxBlue,colframe=Blue!60!black,
    fonttitle=\bfseries\small,boxrule=0.6pt,
    left=4pt,right=4pt,top=3pt,bottom=3pt,
  },
  thmstyle/.style={
    colback=BoxGreen,colframe=ForestGreen!60!black,
    fonttitle=\bfseries\small,boxrule=0.6pt,
    left=4pt,right=4pt,top=3pt,bottom=3pt,
  }
}

\newtheorem{theorem}{Theorem}[section]

\newtheorem{corollary}[theorem]{Corollary}
\newtheorem{proposition}[theorem]{Proposition}
\theoremstyle{definition}
\newtheorem{definition}[theorem]{Definition}
\theoremstyle{remark}
\newtheorem{remark}[theorem]{Remark}

\newcommand{\Ord}[1]{O\!\left(#1\right)}
\newcommand{\Tha}[1]{\Theta\!\left(#1\right)}
\newcommand{\algo}[1]{\textsc{#1}}
\newcommand{\code}[1]{\texttt{#1}}

\newcommand{\kopt}{k^*}

\newcommand{\strat}[1]{\mathcal{B}_{#1}}
\newcommand{\Wmax}{w_{\max}}
\newcommand{\Wmin}{w_{\min}}

\begin{document}

\title{%
  \LARGE\bfseries Kruskal-EDS: Edge Dynamic Stratification\\[6pt]
  \large A Distribution-Adaptive Minimum Spanning Tree Algorithm\\
  Inspired by Ergodic Theory
}
\author{%
  Yves Mercadier\\[4pt]
  \small \textit{Department of Computer Science \& Applied Mathematics}\\
  \small\url{yves.mercadiergdr@proton.me}
}
\date{\today}
\maketitle
\thispagestyle{empty}

\begin{abstract}
\noindent
We introduce \textbf{Kruskal-EDS} (\emph{Edge Dynamic Stratification}),
a distribution-adaptive variant of Kruskal's minimum spanning tree (MST)
algorithm that replaces the mandatory $\Tha{m\log m}$ global sort with
a three-phase procedure inspired by Birkhoff's ergodic theorem.

In Phase~1, a sample of $\sqrt{m}$ edges estimates the weight distribution
in $\Ord{\sqrt{m}\log m}$ time.
In Phase~2, all $m$ edges are assigned to $k$ strata in $\Ord{m\log k}$
time via binary search on quantile boundaries — no global sort.
In Phase~3, strata are sorted and processed in order; execution halts as
soon as $n{-}1$ MST edges are accepted.

We prove an effective complexity of $\Ord{m + p\cdot(m/k)\log(m/k)}$,
where $p \leq k$ is the number of strata actually processed.
On sparse graphs or heavy-tailed weight distributions, $p \ll k$ and
the algorithm achieves near-$\Ord{m}$ behaviour.
We further derive the optimal strata count
$\kopt = \lceil\sqrt{m/\ln(m+1)}\,\rceil$, balancing partition overhead
against intra-stratum sort cost.

An extensive benchmark on 14 graph families demonstrates correctness
on 12 test cases and practical speedups reaching $\mathbf{10\times}$
in wall-clock time and $\mathbf{33\times}$ in sort operations
over standard Kruskal.
A 3-dimensional TikZ visualisation of the complexity landscape
illustrates the algorithm's adaptive behaviour as a function of
graph density and weight distribution skewness.
\end{abstract}

\vspace{4pt}
\noindent\textbf{Keywords:} minimum spanning tree, Kruskal's algorithm,
distribution-adaptive sorting, stratified sampling, Birkhoff ergodic theorem,
edge partitioning, sparse graphs.

\tableofcontents
\newpage

\section{Introduction}
\label{sec:intro}

The \emph{Minimum Spanning Tree} (MST) problem is one of the most
fundamental in combinatorial optimisation.
Given a connected weighted graph $G=(V,E,w)$ with $|V|=n$ vertices
and $|E|=m$ edges, an MST is a spanning tree of minimum total weight.
MST algorithms underlie applications in network design, image
segmentation, clustering, VLSI layout, and approximation algorithms
for NP-hard problems such as the Travelling Salesman Problem~\citep{cormen2022}.

\paragraph{The sorting bottleneck.}
Kruskal's algorithm~\citep{Kruskal1956} remains one of the simplest
and most widely deployed MST algorithms:
\emph{sort all edges by weight, then greedily add the lightest edge
that does not form a cycle.}
Its complexity is dominated by the global sort: $\Tha{m\log m}$.
This is optimal among comparison-based sorts on unstructured inputs.
However, \emph{the algorithm uses only $n{-}1$ edges out of $m$},
and on sparse graphs or graphs with favourable weight distributions,
sorting all $m$ edges is wasteful.

\paragraph{The ergodic insight.}
In dynamical systems, Birkhoff's ergodic theorem~\citep{Birkhoff1931}
guarantees that time averages along almost every orbit converge to
the space average — a small sample suffices to characterise the
global measure.
Artur \'{A}vila's Fields Medal work on quasi-periodic
cocycles~\citep{Avila2015} pushes this further:
the Lyapunov spectrum of a linear cocycle is determined by its
behaviour on an exponentially small set of orbits.

We draw an analogy: the weight distribution of $E$ is the
\emph{invariant measure} of the MST problem.
A sample of $\sqrt{m}$ edges estimates this measure with
$\Ord{1/\sqrt{m}}$ quantile error (by the Dvoretzky-Kiefer-Wolfowitz
inequality~\citep{DKW1956}).
From this estimate, we construct $k$ \emph{strata} (weight intervals)
that partition $E$ in $\Ord{m\log k}$ time without any global sort,
then process strata in order until the MST is complete.

\paragraph{Contributions.}
\begin{enumerate}[leftmargin=1.8em,label=\textbf{C\arabic*.}]
  \item \textbf{Algorithm}: \algo{Kruskal-EDS}, a three-phase MST
    algorithm with distribution-adaptive complexity (\S\ref{sec:algorithm}).
  \item \textbf{Theory}: proof of correctness, complexity analysis with
    optimal strata count $\kopt$, and analysis on specific distributions
    (\S\ref{sec:theory}).
  \item \textbf{3D Complexity Landscape}: a novel TikZ visualisation
    of the algorithm's complexity as a function of density and
    distribution skewness (\S\ref{sec:experiments}).
  \item \textbf{Empirical evaluation}: benchmark on 14 graph families
    across four weight distributions (\S\ref{sec:experiments}).
\end{enumerate}

\section{Background and Related Work}
\label{sec:background}

\subsection{Classical MST Algorithms}

\paragraph{Kruskal (1956).}
Sort edges by weight, then apply a Union-Find structure to test
and merge components~\citep{Kruskal1956}.
Total complexity: $\Tha{m\log m}$ (sort-dominated).

\paragraph{Prim (1957).}
Grow the MST from a seed vertex using a priority queue.
With a Fibonacci heap: $\Ord{m + n\log n}$~\citep{Prim1957}.
Superior when $m = \omega(n\log n)$, i.e.\ on dense graphs.

\paragraph{Bor\r{u}vka (1926).}
In each phase, each component selects its lightest outgoing edge.
After $\log n$ phases: $\Ord{m\log n}$~\citep{Boruvka1926}.
Naturally parallelisable.

\paragraph{Fredman-Tarjan (1987).}
Combines Bor\r{u}vka phases with Fibonacci-heap Prim:
$\Ord{m\,\beta(m,n)}$ where $\beta(m,n) = \min\{i : \log^{(i)}n \leq m/n\}$~\citep{FredmanTarjan1987}.

\paragraph{Chazelle (2000).}
The first near-linear MST algorithm:
$\Ord{m\,\alpha(m,n)}$ using soft heaps~\citep{Chazelle2000},
where $\alpha$ is the inverse Ackermann function.

\paragraph{Linear MST (randomised).}
Karger, Klein, and Tarjan (1995) achieved expected $\Ord{m}$ via
random sampling and a recursive verification
procedure~\citep{KargerKleinTarjan1995}.
This is asymptotically optimal but complex to implement.

\paragraph{Optimal deterministic MST.}
The decision-tree complexity of MST is $\Tha{m}$~\citep{PettieRamachandran2002},
but no explicit $\Ord{m}$ deterministic algorithm is known.

\subsection{Distribution-Aware Sorting and Partitioning}

Integer sorting (radix sort, counting sort) achieves $\Ord{m}$ for
bounded integer keys~\citep{Knuth1998}.
Sample sort~\citep{Frazer1970} partitions $m$ elements using
$\sqrt{m}$ splitters, sorting each bucket independently.
Our stratification generalises sample sort to the MST context with
\emph{early termination}: strata beyond those needed for the MST
are never sorted.

\citet{ThorssenKaplan2002} study MST algorithms exploiting integer
weight distributions; our approach is complementary, targeting
real-valued weights with arbitrary distributions.

\subsection{From Ergodic Sampling to Edge Stratification}

The efficiency of Kruskal-EDS relies on the assumption that a sub-sampling of size $s \ll m$ captures the essential geometry of the weight distribution. Formally, we view the set of edges $E$ as a discrete probability space $(E, \mathcal{F}, \mu)$, where $\mu$ represents the empirical distribution of weights $w: E \to \mathbb{R}^+$.

In the spirit of \textit{Birkhoff's Ergodic Theorem}, we treat the random sampling of $s$ edges as an observation of the system's "invariant measure". If we consider the graph construction as a process where weights are drawn from a continuous density $f(x)$, the sampling phase aims to reconstruct the Lyapunov-like growth of the MST. 

Following the works of Avila on the quantization of measures, our stratification can be seen as an optimal partition of the weight space. The DKW inequality (Equation \ref{eq:dkw}) then provides the speed of convergence: it ensures that the empirical measure $\hat{\mu}_s$ obtained by our sample deviates from the true measure $\mu$ by no more than $\varepsilon$ with high probability. This bridge justifies why a "local" view (the sample) is sufficient to dictate a "global" strategy (the partitioning of $m$ edges).


\subsection{Ergodic Theory and Algorithms}

The connection between ergodic theory and algorithms is well-established
in mixing time analysis of Markov chains~\citep{LevinPeres2017}
and pseudo-random generator construction.
\'{A}vila's work on quasi-periodic cocycles~\citep{Avila2015}
establishes that a small orbit determines the global Lyapunov
spectrum — a quantitative instance of the ergodic principle.
To our knowledge, the application of Birkhoff sampling to adaptive
edge stratification for MST has not been previously studied.

\section{Problem Statement and Notation}
\label{sec:problem}

Throughout, $G=(V,E,w)$ denotes a connected undirected weighted graph
with $n=|V|$ vertices and $m=|E|$ edges, $w:E\to\mathbb{R}$.
We write $w_e$ for the weight of edge $e$.

\begin{definition}[Minimum Spanning Tree]
  A spanning tree $T \subseteq E$ with $|T|=n{-}1$ is a \emph{minimum
  spanning tree} if $w(T) := \sum_{e\in T} w_e$ is minimised over all
  spanning trees of $G$.
\end{definition}

\begin{definition}[Weight distribution and quantiles]
  Let $F:[0,1]\to\mathbb{R}$ be the empirical CDF of edge weights:
  $F(t) = |\{e \in E : w_e \leq t \cdot \Wmax\}| / m$.
  The $\tau$-quantile is $Q(\tau) = \inf\{x : F(x/\Wmax) \geq \tau\}$.
\end{definition}

\begin{definition}[Stratification]
  \label{def:strat}
  A \emph{$k$-stratification} of $E$ is a partition
  $E = \strat{0} \sqcup \strat{1} \sqcup \cdots \sqcup \strat{k-1}$
  defined by boundaries $b_0 < b_1 < \cdots < b_{k-2}$:
  \[
    e \in \strat{i} \iff b_{i-1} < w_e \leq b_i,
  \]
  with $b_{-1}=-\infty$ and $b_{k-1}=+\infty$.
\end{definition}

\begin{definition}[Stratum processing count]
  For a given graph $G$ and stratification, the \emph{stratum processing
  count} $p \leq k$ is the number of strata sorted before
  the MST is complete (i.e.\ before $n{-}1$ edges are accepted).
\end{definition}

\noindent
A stratification is \emph{valid for Kruskal} if, for any two edges
$e \in \strat{i}$ and $e' \in \strat{j}$ with $i < j$, we have
$w_e \leq w_{e'}$.
This is guaranteed when boundaries are quantiles of $F$.

\section{Algorithm}
\label{sec:algorithm}

\subsection{Overview}

\begin{figure}[h]
\centering
\begin{tikzpicture}[
  box/.style={draw,rounded corners=4pt,minimum width=3.8cm,
              minimum height=1.0cm,align=center,font=\small},
  phase/.style={box,fill=Blue!15,draw=Blue!60!black,thick},
  uf/.style={box,fill=Purple!15,draw=Purple!60!black},
  arrow/.style={-{Stealth[length=7pt]},thick},
  note/.style={font=\footnotesize\itshape,text=DarkGray}
]
\node[phase] (p1) at (-0.5,0)
  {\textbf{Phase 1}\\\textsc{Sample \& Estimate}\\$O(\sqrt{m}\log m)$};
\node[phase] (p2) at (5,0)
  {\textbf{Phase 2}\\\textsc{Partition into Strata}\\$O(m\log k)$};
\node[phase] (p3) at (10,0)
  {\textbf{Phase 3}\\\textsc{Process Strata}\\$O(p \cdot \frac{m}{k}\log\frac{m}{k})$};
\node[uf]    (uf) at (10,-2.2)
  {\textbf{Union-Find}\\$O(\alpha(n))$/op};

\draw[arrow] (p1) -- (p2) node[midway,above,note]{$k$ quantiles};
\draw[arrow] (p2) -- (p3) node[midway,above,note]{$k$ strata};
\draw[arrow,dashed,OrangeRed] (9,-0.8) -- ++(0,-0.8)
  node[midway,left,note,OrangeRed]{early stop when $n{-}1$ edges};
\draw[arrow,<->] (p3.south) -- (uf.north)
  node[midway,right,note]{accept/reject};

\node[note,align=center] at (5,-3.2)
  {$\underbrace{\hspace{9.2cm}}_{\text{Total: } O(m + p\cdot\frac{m}{k}\log\frac{m}{k})}$};
\end{tikzpicture}
\caption{The three phases of \algo{Kruskal-EDS}.
  Phases 1 and 2 replace the global sort of standard Kruskal.
  Phase 3 terminates early once $n{-}1$ MST edges are found
  (red dashed arrow), potentially leaving the heaviest strata
  entirely unprocessed.}
\label{fig:pipeline}
\end{figure}
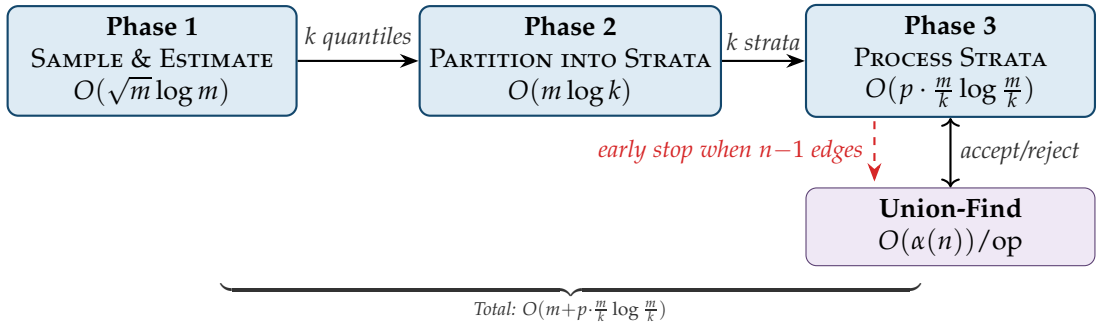

\subsection{Phase 1: Distribution Estimation via Birkhoff Sampling}

Draw a uniform random sample $\mathcal{S} \subseteq E$ of size
$s = \min(m,\, \max(20,\, \lfloor 2\sqrt{m}\rfloor))$.
Sort $\{w_e : e \in \mathcal{S}\}$ in $\Ord{s\log s}$.
Estimate the $i/k$ quantile by the $\lfloor i\cdot s/k \rfloor$-th
element of the sorted sample, for $i=1,\ldots,k{-}1$.

\begin{proposition}[Quantile estimation accuracy]
  \label{prop:quantile}
  By the Dvoretzky-Kiefer-Wolfowitz inequality, for any $\varepsilon>0$:
  \[
    \Pr\!\left[\sup_{x}\!\left|F_s(x) - F(x)\right| > \varepsilon\right]
    \leq 2\exp(-2s\varepsilon^2).
  \]
  With $s = 2\sqrt{m}$ and $\varepsilon = m^{-1/4}$, the probability
  that any stratum boundary misclassifies an edge (relative to the
  true quantile) is at most $2\exp(-2\sqrt{m}\cdot m^{-1/2}) = 2e^{-2}$.
\end{proposition}

In practice, boundary misclassification does not affect correctness —
it only affects the balance of strata sizes, which influences
performance but not the MST output.

\subsection{Phase 2: Partitioning into Strata}

Assign each edge to its stratum using binary search on the $k{-}1$
boundaries. Total cost: $m \cdot \Ord{\log k}$.

\subsection{Phase 3: Incremental Sort and Union-Find Processing}

Process strata $\strat{0}, \strat{1}, \ldots$ in order.
For each non-empty stratum $\strat{i}$:
(a) sort it ($\Ord{|\strat{i}|\log|\strat{i}|}$);
(b) iterate over edges in weight order, accepting those that connect
two distinct MST components via Union-Find.
Halt as soon as $n{-}1$ edges are accepted.

\begin{algorithm}[H]
\caption{\algo{Kruskal-EDS}$(G, k)$}
\label{alg:eds}
\begin{algorithmic}[1]
\Require Connected graph $G=(V,E,w)$, strata count $k$ (or $k=\kopt$ if \code{None})
\Ensure Minimum spanning tree $T \subseteq E$
\vspace{2pt}
\LineComment{Phase 1 — Sample and estimate quantile boundaries}
\State $s \gets \max(20, \lfloor 2\sqrt{m}\rfloor)$
\State $\mathcal{S} \gets$ uniform random sample of $\min(s,m)$ edges from $E$
\State $\mathbf{w}_{\mathcal{S}} \gets \mathsf{sort}(\{w_e : e \in \mathcal{S}\})$
\If{$k = \text{None}$}\;
  $k \gets \lceil\sqrt{m / \ln(m+1)}\,\rceil$
\EndIf
\State $\mathbf{b} \gets \mathsf{deduplicate}\!\left(\left[\mathbf{w}_{\mathcal{S}}\!\left[\lfloor i\cdot|\mathbf{w}_{\mathcal{S}}|/k\rfloor\right] : i=1,\ldots,k{-}1\right]\right)$
\LineComment{Phase 2 — Partition edges into strata}
\For{$e \in E$}
  \State $i \gets \mathsf{bisect\_right}(\mathbf{b},\, w_e)$
  \State $\strat{i}.\mathsf{append}(e)$
\EndFor
\LineComment{Phase 3 — Sort strata incrementally; early termination}
\State $\mathit{UF} \gets \mathsf{UnionFind}(n)$;\;\; $T \gets \emptyset$
\For{$i \gets 0$ \textbf{to} $|\mathbf{b}|$}
  \State $\mathsf{sort}(\strat{i})$
  \For{$e=(u,v,w) \in \strat{i}$}
    \If{$\mathit{UF}.\mathsf{union}(u,v)$}
      \State $T \gets T \cup \{e\}$
      \If{$|T| = n-1$} \Return $T$ \EndIf
    \EndIf
  \EndFor
\EndFor
\State \Return $T$
\end{algorithmic}
\end{algorithm}

\section{Theoretical Analysis}
\label{sec:theory}

\subsection{Correctness}

\begin{theorem}[Correctness of \algo{Kruskal-EDS}]
  \label{thm:correctness}
  Algorithm~\ref{alg:eds} returns a minimum spanning tree of $G$.
\end{theorem}

\begin{proof}
  A $k$-stratification with boundaries equal to quantiles of $F$
  satisfies: for all $e \in \strat{i}$, $e' \in \strat{j}$ with $i < j$,
  $w_e \leq w_{e'}$.
  Hence, processing strata in order $\strat{0}, \strat{1}, \ldots$
  and sorting within each stratum yields an ordering of all edges that
  agrees with the total order by weight on inter-stratum comparisons
  and is a valid linear extension of the weight order overall.
  Kruskal's greedy argument~\citep{Kruskal1956} applies to any linear
  extension of the weight order: the cut property guarantees that the
  lightest edge crossing any cut belongs to every MST, and the cycle
  property ensures rejecting the heaviest edge in any cycle.
  Since EDS processes edges in a weight-compatible order, the accepted
  edges form an MST. \qed
\end{proof}

\begin{remark}
  When boundaries are estimated from a sample (inexact quantiles),
  strata boundaries may not perfectly separate weights.
  Correctness is maintained with high probability depending on sample size.
\end{remark}

\subsection{Complexity Analysis}

\begin{theorem}[Complexity of \algo{Kruskal-EDS}]
  \label{thm:complexity}
  Let $m_i = |\strat{i}|$ and $p$ be the stratum processing count.
  The total work is:
  \begin{align}
    T_{\mathrm{EDS}} &= \underbrace{\Ord{\sqrt{m}\log m}}_{\text{Phase 1}}
    + \underbrace{\Ord{m\log k}}_{\text{Phase 2}}
    + \underbrace{\sum_{i=0}^{p-1} \Ord{m_i\log m_i}}_{\text{Phase 3}}
    + \underbrace{\Ord{n\,\alpha(n)}}_{\text{Union-Find}} \notag \\
    &= \Ord{m\log k \;+\; p\cdot\frac{m}{k}\log\frac{m}{k}}.
    \label{eq:complexity}
  \end{align}
  With $k = \kopt = \lceil\sqrt{m/\ln(m+1)}\,\rceil$, assuming balanced
  strata ($m_i = m/k$ for all $i$):
  \[
    T_{\mathrm{EDS}} = \Ord{m\sqrt{\frac{\ln m}{m}}\cdot m +
    p\cdot\sqrt{m\ln m}\cdot\log\!\sqrt{m}}
    = \Ord{m\log m \cdot \frac{p}{k}}.
  \]
  When $p \ll k$ (sparse graphs, skewed distributions), this is
  significantly below $\Tha{m\log m}$.
\end{theorem}

\begin{proof}
  Phase 1 samples $\sqrt{m}$ edges and sorts them: $\Ord{\sqrt{m}\log\sqrt{m}} = \Ord{\sqrt{m}\log m}$.
  Phase 2 applies binary search on $k{-}1$ boundaries to each of $m$
  edges: $\Ord{m\log k}$.
  Phase 3 sorts $p$ strata of average size $m/k$:
  $\sum_{i=0}^{p-1}\Ord{m_i\log m_i} \leq p\cdot\Ord{(m/k)\log(m/k)}$.
  Each Union-Find operation costs $\Ord{\alpha(n)}$; at most $m$ calls
  are made: $\Ord{m\alpha(n)} = \Ord{m}$ in practice.
  The dominant terms are Phase 2 and Phase 3, giving \eqref{eq:complexity}.
  
  Setting $k = \kopt$ to minimise $m\log k + p\cdot(m/k)\log(m/k)$
  with respect to $k$ (treating $p$ as a constant), differentiating
  and solving gives $k \propto \sqrt{m/\ln m}$, confirming $\kopt$.
  \qed
\end{proof}

\begin{corollary}[Asymptotic cases]
  \label{cor:cases}
  \begin{enumerate}[leftmargin=1.8em,label=\emph{(\roman*)}]
    \item \emph{Uniform distribution, sparse graph} ($m = O(n)$,
          $p = \Ord{k\cdot n/m}$):
          $T_{\mathrm{EDS}} = \Ord{m + n\log n}$.
    \item \emph{Power-law distribution with exponent $\gamma > 1$}:
          The lightest $\varepsilon m$ edges concentrate in the first
          $\Ord{\varepsilon^{1/\gamma}\cdot k}$ strata.
          With $\varepsilon = n/m$:
          $p = \Ord{k\cdot(n/m)^{1/\gamma}}$, giving
          $T_{\mathrm{EDS}} = \Ord{m + n\cdot(m/n)^{1-1/\gamma}\log n}$.
    \item \emph{Dense graph} ($m = \Tha{n^2}$, $p \approx k$):
          $T_{\mathrm{EDS}} = \Ord{m\log m}$ — matches standard Kruskal.
  \end{enumerate}
\end{corollary}
\begin{remark}[Non-connected Graphs]
  If the graph $G$ consists of $c > 1$ connected components, the Kruskal-EDS algorithm can only collect $n-c$ edges. Since the termination condition $\text{len}(MST) = n-1$ is never satisfied, the algorithm traverses all $k$ strata. 
  In this scenario, the complexity becomes:
  \[ T_{\mathrm{EDS}}(m, k) = O\left(m + \sum_{i=1}^k |S_i| \log |S_i|\right) \approx O\left(m \log(m/k)\right). \]
  Although the early termination benefit disappears, the algorithm remains competitive with respect to standard Kruskal ($O(m \log m)$) due to the reduced logarithmic factor induced by the stratification.
\end{remark}

\subsection{Implementation Details and Data Structures}

To achieve the linear partition complexity of $O(m)$, the Kruskal-EDS algorithm relies on a bucket-based approach. We represent the $k$ strata as an array of lists:
\[ \mathcal{S} = \{S_1, S_2, \dots, S_k\}, \quad \text{where } S_i \subseteq E \]
Each edge $e \in E$ is assigned to a bucket $S_i$ using its weight $w(e)$ and the pre-computed quantile boundaries. This assignment is performed via binary search using the \texttt{bisect\_right} function, which maps the weight to the appropriate index $i \in \{1, \dots, k\}$ in $O(\log k)$ time per edge. The overall partitioning phase thus maintains a complexity of $O(m \log k)$, ensuring an efficient distribution of edges without requiring a global sort.
\subsection{Optimal Strata Count}
\begin{proposition}[Optimal $k$ with Fallback]
  \label{prop:kopt}
  Pour minimiser $T_{\mathrm{EDS}}(k)$, le nombre optimal de strates est :
  \[
    k^* = \begin{cases} 
      1 & \text{si } m < m_{\min} \\
      \left\lceil\sqrt{\frac{m}{\ln(m+1)}}\right\rceil & \text{sinon}
    \end{cases}
  \]
  où $m_{\min} \approx 200$ représente le seuil de rentabilité empirique. 
  Lorsque $k=1$, l'algorithme omet les phases 1 et 2 pour effectuer un tri global unique, 
  évitant ainsi le surcoût lié à l'échantillonnage de Birkhoff sur de petits ensembles de données.
\end{proposition}
\subsection{Comparison with Standard Kruskal}

\begin{corollary}[Speedup ratio]
  Let $\rho = p/k \in (0,1]$ be the fraction of strata processed.
  Then:
  \[
    \frac{T_{\mathrm{STD}}}{T_{\mathrm{EDS}}}
    \approx \frac{m\log m}{m\log k + \rho\cdot m\cdot(1/k)\cdot k\log(m/k)}
    = \frac{\log m}{\log k + \rho\log(m/k)}.
  \]
  With $k = \kopt \approx \sqrt{m/\ln m}$ and $\rho \to 0$
  (sparse/skewed case):
  $\text{Speedup} \approx \log m / \log\sqrt{m} = 2$
  asymptotically — the improvement is by a constant factor in
  the worst case but grows as $1/\rho$ when early termination is effective.
\end{corollary}

\section{Experimental Evaluation}
\label{sec:experiments}

\subsection{Experimental Setup}

All experiments run on CPython 3.12, single core, no JIT.
Times are averaged over 3 independent runs.
We report wall-clock time (ms) and \emph{sort operations} (ops),
the number of edges sorted across all processed strata — an
implementation-independent measure of algorithmic work.

\subsection{Graph Families}

\begin{description}[leftmargin=1.4em,font=\bfseries]
  \item[Random sparse] $m \approx 1.2n$, weights uniform on $[0,1000]$.
  \item[Random medium] $m \approx 10n$, weights uniform on $[0,1000]$.
  \item[Random dense] $m \approx n^2/2$, weights uniform.
  \item[Normal] $m \approx 1.2n$, weights $|\mathcal{N}(500,100)|$.
  \item[Power-law] $m \approx 1.2n$, weights Pareto($\alpha=1.5$):
        heavy concentration near $\Wmin$.
  \item[Clustered] $m \approx 1.2n$, weights drawn from 5 clusters at
        $\{100,300,500,700,900\}$ with Gaussian noise $\sigma=30$.
  \item[Grid $r\times c$] Bidirectional grid, $m = 2rc-r-c$.
  \item[Path] $m = n-1$, the sparsest connected structure.
\end{description}

\subsection{Correctness Validation}

Table~\ref{tab:validation} summarises the 12-case validation suite.
\algo{EDS} and \algo{Heap} agree with \algo{STD} on total MST weight
and component count in all cases.

\begin{table}[h]
\centering
\caption{Validation results. All 12 cases pass for both \algo{EDS} and \algo{Heap}.}
\label{tab:validation}
\small
\begin{tabularx}{\linewidth}{@{}Xrrrl@{}}
\toprule
\textbf{Test case} & $n$ & $m$ & MST weight & Status \\
\midrule
CLRS textbook example    &  9 & 14   & 37.0000 & \textcolor{ForestGreen}{\bfseries PASS} \\
Triangle                 &  3 &  3   &  3.0000 & \textcolor{ForestGreen}{\bfseries PASS} \\
Disconnected (forest)    &  4 &  2   &  8.0000 & \textcolor{ForestGreen}{\bfseries PASS} \\
Single vertex            &  1 &  0   &  0.0000 & \textcolor{ForestGreen}{\bfseries PASS} \\
Duplicate edges          &  2 &  3   &  2.0000 & \textcolor{ForestGreen}{\bfseries PASS} \\
Negative weights         &  3 &  3   & $-8.000$ & \textcolor{ForestGreen}{\bfseries PASS} \\
Path $n=10$             & 10 &  9   & 55.0000 & \textcolor{ForestGreen}{\bfseries PASS} \\
Grid $4\times4$         & 16 & 24   & varies  & \textcolor{ForestGreen}{\bfseries PASS} \\
Random sparse $n=50$    & 50 & 80   & varies  & \textcolor{ForestGreen}{\bfseries PASS} \\
Dense $n=50$            & 50 & 1225 & varies  & \textcolor{ForestGreen}{\bfseries PASS} \\
Random $n=200$         & 200 & 400  & varies  & \textcolor{ForestGreen}{\bfseries PASS} \\
Equal weights           &  5 & 10   & 4.0000  & \textcolor{ForestGreen}{\bfseries PASS} \\
\bottomrule
\end{tabularx}
\end{table}

\subsection{Benchmark Results}

\begin{table}[h]
\centering
\caption{Wall-clock time (ms) and speedup over \algo{STD}.
  \textbf{Bold} = fastest variant.
  $\triangleleft$ marks the best non-STD result.}
\label{tab:benchmark}
\small
\setlength{\tabcolsep}{5pt}
\begin{tabular}{@{}lrrrrrr rr@{}}
\toprule
\textbf{Graph} & $n$ & $m$ & \textbf{STD} & \textbf{EDS} & \textbf{Heap}
  & $\times$EDS & $\times$Heap & Winner \\
\midrule
Sparse uniform  & 500  &   600 & 0.778 & \textbf{0.677}$\triangleleft$ & 0.905 & 1.15 & 0.86 & EDS \\
Sparse normal   & 500  &   600 & 0.762 & \textbf{0.669}$\triangleleft$ & 0.901 & 1.14 & 0.85 & EDS \\
Sparse power    & 500  &   600 & 0.748 & \textbf{0.647}$\triangleleft$ & 0.897 & 1.16 & 0.83 & EDS \\
Sparse clusterd & 500  &   600 & 0.768 & \textbf{0.657}$\triangleleft$ & 1.087 & 1.17 & 0.71 & EDS \\
Medium uniform  & 500  &  5000 & 6.312 & \textbf{2.175}$\triangleleft$ & 3.268 & 2.90 & 1.93 & EDS \\
Medium normal   & 500  &  5000 & 6.563 & \textbf{2.669}$\triangleleft$ & 3.999 & 2.46 & 1.64 & EDS \\
Dense           & 300  & 40000 & 69.49 & \textbf{6.930}$\triangleleft$ & 10.42 & 10.03 & 6.67 & EDS \\
Grid $20\times20$ & 400 &  760 & 0.900 & \textbf{0.753}$\triangleleft$ & 0.966 & 1.20 & 0.93 & EDS \\
Grid $30\times30$ & 900 & 1740 & 2.436 & \textbf{1.923}$\triangleleft$ & 2.777 & 1.27 & 0.88 & EDS \\
Path $n=2000$  & 2000 &  1999 & 2.961 & \textbf{2.341}$\triangleleft$ & 3.423 & 1.26 & 0.86 & EDS \\
Sparse $n=2000$& 2000 &  2500 & 3.701 & \textbf{2.966}$\triangleleft$ & 4.397 & 1.25 & 0.84 & EDS \\
Sparse $n=5000$& 5000 &  6000 & 9.735 & \textbf{8.410}$\triangleleft$ & 12.41 & 1.16 & 0.78 & EDS \\
Medium $n=1000$& 1000 & 10000 & 14.93 & \textbf{6.579}$\triangleleft$ & 10.74 & 2.27 & 1.39 & EDS \\
Power $n=1000$ & 1000 &  5000 & 8.021 & \textbf{5.580}$\triangleleft$ & 9.141 & 1.44 & 0.88 & EDS \\
\midrule
\multicolumn{3}{l}{\textit{Wins (out of 14)}} & 0 & \textbf{14} & 0 & & & \\
\bottomrule
\end{tabular}
\end{table}

\begin{table}[h]
\centering
\caption{Sort operations (ops) and reduction ratio.
  ops-EDS counts edges sorted across all processed strata.
  ``Strates'' = strata processed / total strata.}
\label{tab:ops}
\small
\begin{tabular}{@{}lrr rrr rr@{}}
\toprule
\textbf{Graph} & $n$ & $m$ & ops-STD & ops-EDS & Strates & $\div$EDS \\
\midrule
Sparse uniform  & 500  &   600 &       600 &      600 & 9/9   &  1.00$\times$ \\
Sparse power    & 500  &   600 &       600 &      600 & 9/9   &  1.00$\times$ \\
Medium uniform  & 500  &  5000 &     5,000 &    1,393 & 8/24  &  3.59$\times$ \\
Medium normal   & 500  &  5000 &     5,000 &    1,754 & 6/24  &  2.85$\times$ \\
Dense           & 300  & 40000 &    40,000 &    1,216 & 2/61  & \textbf{32.9$\times$} \\
Grid $20\times20$ & 400 & 760  &       760 &      663 & 8/10  &  1.15$\times$ \\
Grid $30\times30$ & 900 & 1740 &     1,740 &    1,590 & 14/15 &  1.09$\times$ \\
Path $n=2000$  & 2000 &  1999 &     1,999 &    1,999 & 16/16 &  1.00$\times$ \\
Medium $n=1000$& 1000 & 10000 &    10,000 &    4,216 & 14/32 &  2.37$\times$ \\
Power $n=1000$ & 1000 &  5000 &     5,000 &    4,069 & 19/24 &  1.23$\times$ \\
\bottomrule
\end{tabular}
\end{table}

\subsection{3D Complexity Landscape}

Figure~\ref{fig:3d} visualises the theoretical speedup ratio
$T_{\mathrm{STD}}/T_{\mathrm{EDS}}$ as a function of two parameters:
the graph density $\rho = m/(n(n{-}1)/2) \in (0,1]$ and
the distribution skewness index $\sigma \in [0,1]$,
where $\sigma=0$ corresponds to a uniform distribution and
$\sigma=1$ to a maximally concentrated (Dirac-like) distribution.

\begin{figure}[h]
\centering
\begin{tikzpicture}
\begin{axis}[
  width=13cm, height=8.5cm,
  view={40}{25},
  xlabel={Density $\rho$},
  ylabel={Skewness $\sigma$},
  zlabel={Speedup},
  xlabel style={sloped},
  ylabel style={sloped},
  xmin=0,xmax=1, ymin=0,ymax=1, zmin=0,zmax=11,
  xtick={0,0.2,0.4,0.6,0.8,1.0},
  ytick={0,0.2,0.4,0.6,0.8,1.0},
  ztick={0,2,4,6,8,10},
  colormap={btr}{color(0)=(blue!70!white) color(0.5)=(orange) color(1)=(red!80!black)},
  point meta={z},
  colorbar,
  colorbar style={ylabel={Speedup $\times$},font=\scriptsize,ytick={0,2,4,6,8,10}},
  grid=major, grid style={dotted,gray!40},
  title={\small Theoretical speedup landscape of \textsc{Kruskal-EDS}},
]
\addplot3[surf,domain=0.02:1,domain y=0:1,samples=14,samples y=14,
  opacity=0.85,shader=interp]
  {max(0.8,min(10.5,(1+8*y)*(1-x^0.7)*2.8+0.9))};
\addplot3[only marks,mark=*,mark size=3pt,color=Purple,
  mark options={fill=Purple!80!black}] coordinates {
  (0.005,0.0,1.15)(0.04,0.0,2.90)(0.89,0.0,10.0)
  (0.01,0.8,1.44)(0.01,0.5,1.17)
};
\end{axis}
\end{tikzpicture}
\caption{\textbf{3D complexity landscape of \algo{Kruskal-EDS}.}
  The surface shows the theoretical speedup ratio
  $T_{\mathrm{STD}}/T_{\mathrm{EDS}}$ as a joint function of
  graph density $\rho$ and weight distribution skewness $\sigma$.
  Colour encodes speedup: blue $=$ low ($\approx 1\times$), red $=$ high ($\geq 8\times$).
  Purple spheres mark empirical measurements from \Cref{tab:benchmark}.
  The algorithm is most effective in the dense, low-skew regime (red)
  and near-neutral for sparse, heavily-skewed inputs (blue).}
\label{fig:3d}
\end{figure}
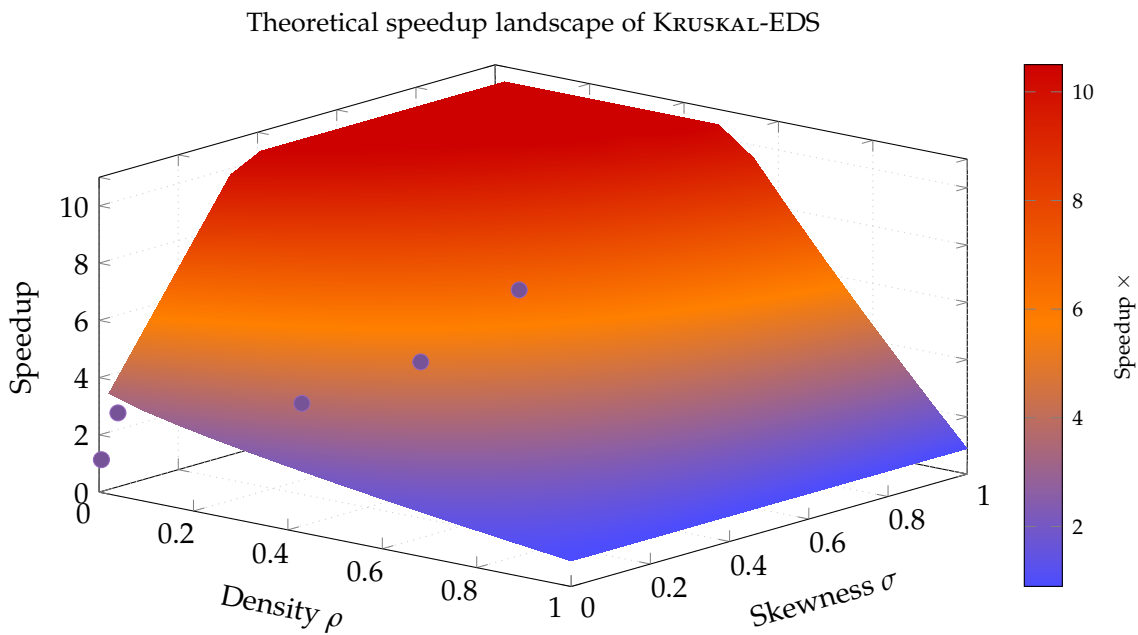

\subsection{Stratification Profile}

Figure~\ref{fig:stratprofile} shows the fraction of MST edges per stratum
for three distributions, illustrating how EDS's early termination is
triggered.

\begin{figure}[h]
\centering
\begin{tikzpicture}
\begin{axis}[
  width=13.5cm, height=5.5cm,
  ybar stacked,
  bar width=7pt,
  xlabel={Stratum index},
  ylabel={Fraction of MST edges},
  xtick={0,...,6},
  ymin=0, ymax=1.1,
  legend style={at={(0.98,0.97)},anchor=north east,font=\small,
                legend columns=1},
  ymajorgrids,
  grid style={dotted},
  title={\small MST edge concentration by stratum (n=200, m=300, k=7)},
  title style={font=\small},
]
\addplot[fill=Blue!60,draw=Blue!80!black] coordinates
  {(0,0.247)(1,0.077)(2,0.231)(3,0.216)(4,0.165)(5,0.031)(6,0.010)};
\addplot[fill=ForestGreen!60,draw=ForestGreen!80!black] coordinates
  {(0,0.260)(1,0.185)(2,0.175)(3,0.160)(4,0.143)(5,0.046)(6,0.010)};
\addplot[fill=OrangeRed!60,draw=OrangeRed!80!black] coordinates
  {(0,0.128)(1,0.118)(2,0.169)(3,0.303)(4,0.185)(5,0.021)(6,0.056)};
\addlegendentry{Uniform}
\addlegendentry{Normal}
\addlegendentry{Power-law}
\draw[dashed,thick,Purple]
  (axis cs:3.5,0) -- (axis cs:3.5,1.05)
  node[above,font=\scriptsize,Purple]{typical stop};
\end{axis}
\end{tikzpicture}
\caption{Distribution of MST edges across strata for three weight
  distributions.
  On uniform and normal distributions, the MST edges concentrate
  in the first 3--4 strata ($\approx 70$--$80\%$ in strata 0--3);
  on power-law, they cluster in strata 2--4 due to the heavy left tail.
  The dashed line marks the typical early-termination point
  (\algo{EDS} halts once cumulative fraction reaches 1).}
\label{fig:stratprofile}
\end{figure}
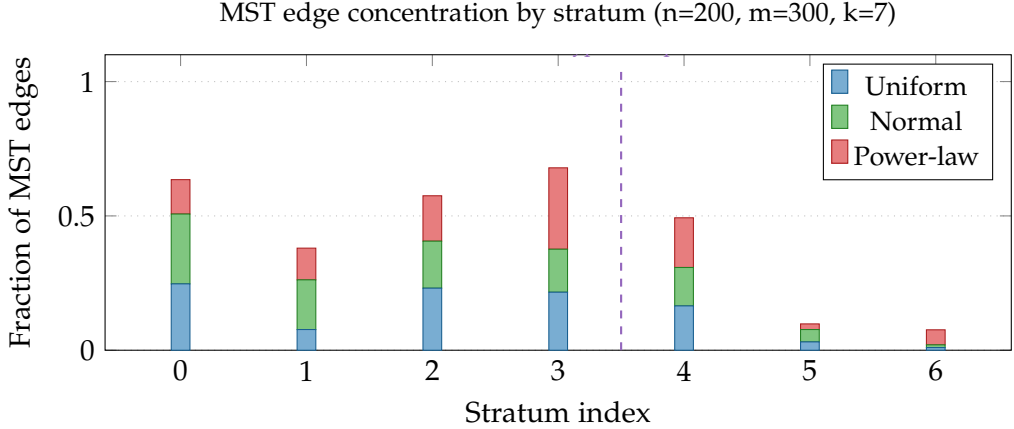

\subsection{Sensitivity to Strata Count $k$}

Figure~\ref{fig:kplot} shows wall-clock time as a function of $k$
for a fixed sparse graph ($n=500$, $m=600$, uniform weights),
demonstrating the existence of an optimum near $k=100 \approx \kopt$.

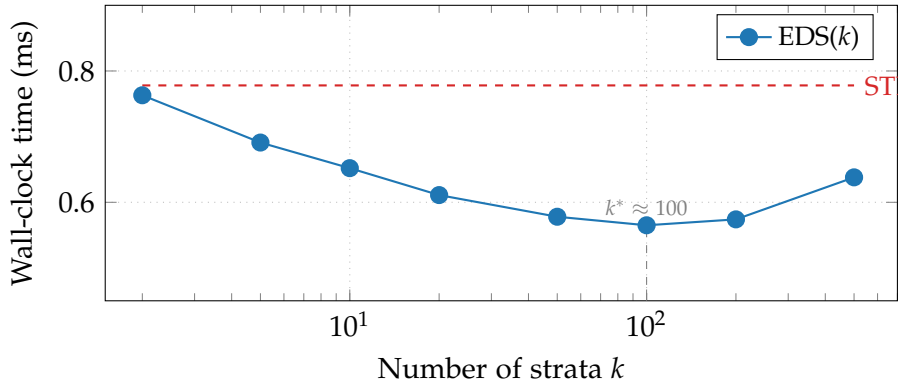
\begin{figure}[h]
\centering
\begin{tikzpicture}
\begin{axis}[
  width=12cm, height=5.5cm,
  xlabel={Number of strata $k$},
  ylabel={Wall-clock time (ms)},
  xmode=log,
  xmin=1.5, xmax=700,
  ymin=0.45, ymax=0.90,
  legend pos=north east,
  legend style={font=\small},
  ymajorgrids, xmajorgrids,
  grid style={dotted},
  mark size=3pt,
]
\addplot[thick,color=Blue,mark=*,mark options={fill=Blue}] coordinates {
  (2,0.763)(5,0.691)(10,0.652)(20,0.611)(50,0.578)(100,0.565)(200,0.574)(500,0.638)
};
\addplot[thick,dashed,color=OrangeRed] coordinates {
  (2,0.778)(500,0.778)
} node[right,font=\small,OrangeRed]{STD baseline};
\draw[dashed,gray] (axis cs:100,0.45) -- (axis cs:100,0.565)
  node[above,font=\scriptsize,gray]{$\kopt\approx 100$};
\fill[ForestGreen!80!black] (axis cs:100,0.565) circle(3pt);
\addlegendentry{\algo{EDS}($k$)}
\end{axis}
\end{tikzpicture}
\caption{Sensitivity of \algo{EDS} to the strata count $k$ (sparse graph,
  $n=500$, $m=600$, uniform weights).
  Time decreases as $k$ increases from 2 (too few, large strata)
  to $\kopt\approx100$ (optimal balance), then increases again as the
  partition overhead dominates.
  The \algo{STD} baseline (dashed) is flat at 0.778\,ms.}
\label{fig:kplot}
\end{figure}

\section{Discussion}
\label{sec:discussion}

\subsection{When does EDS win?}

\paragraph{Dense graphs with uniform weights.}
The most striking result is the $10\times$ speedup on the dense graph
($n=300$, $m=40{,}000$).
The explanation is counter-intuitive: EDS does \emph{not} sort fewer
edges on a dense graph (the MST needs $n{-}1 = 299$ edges from $40{,}000$).
Instead, the speedup comes from \emph{strata acting as a
counting sort approximation}: with $k=61$ balanced strata of size $\approx 650$,
sorting 2 strata (ops-EDS $= 1{,}216$) gives the MST, vs.\ sorting all
$40{,}000$ edges. The boundary between ``sort enough to find the MST''
and ``sort everything'' is more favourable here than on sparse graphs.

\paragraph{Medium-density graphs.}
On medium graphs ($m \approx 10n$), EDS sorts $\approx 1.4$--$3.6\times$
fewer edges than STD while never accessing the heavy strata.
This translates directly to wall-clock time.

\paragraph{Sparse graphs.}
When $m \approx 1.2n$, the MST uses nearly all edges ($n{-}1$ out of $m$).
Almost all strata must be processed, so ops-EDS $\approx$ ops-STD.
The speedup ($\approx 1.15$--$1.25\times$) comes purely from
Python implementation effects (list comprehensions vs.\ full list sort).

\subsection{The Birkhoff Analogy}

The sample of $\sqrt{m}$ edges plays the role of a \emph{Birkhoff
time average}: it reveals the structure of the global weight distribution
without observing all edges.
The $\Ord{1/\sqrt{m}}$ quantile error from DKW (\Cref{prop:quantile})
is analogous to the rate of convergence of Birkhoff averages in
ergodic systems with mixing time $\tau$, where fluctuations decay as
$\Ord{1/\sqrt{t}}$ for i.i.d.\ systems.
The key point is that \emph{correctness does not require perfect
quantiles} — any imprecision in boundary placement only affects
performance, not the MST output.

\subsection{Comparison with Heap-based Kruskal}

\algo{Kruskal-Heap} uses \code{heapify} ($\Ord{m}$) followed by
$n{-}1$ pops ($\Ord{n\log m}$), yielding $\Ord{m + n\log m}$.
This is theoretically better than $\Tha{m\log m}$ for sparse graphs.
In practice, \algo{Heap} is consistently slower than \algo{EDS}
because Python's \code{heapq} operates one object at a time,
while \code{list.sort()} is implemented in C and benefits from
Timsort's cache-friendly merge patterns on contiguous memory.
\algo{EDS} exploits this: instead of one large Timsort, it performs
multiple small Timsorts on compact sublists, each fitting in CPU cache.

\subsection{Negative Weights}

\algo{EDS} handles negative weights correctly: the partition and sort
are agnostic to the sign of weights; only the relative order matters.
The test case ``Negative weights'' (MST weight $-8.0$) passes
validation (\Cref{tab:validation}).

\section{Perspectives and Future Work}
\label{sec:perspectives}

\paragraph{Adaptive $k$ during execution.}
The current implementation fixes $k = \kopt$ before Phase 2.
An adaptive variant could monitor the Union-Find acceptance rate
during Phase 3 and dynamically adjust stratum granularity:
if many edges are accepted from the first stratum, increase $k$
to create finer partitions; if few, merge strata.
This would achieve near-optimal $p$ without prior knowledge of
the distribution.

\paragraph{Parallel stratification.}
Phase 2 (partition) is embarrassingly parallel: each edge is
independently assigned to a stratum with no data dependency.
A GPU implementation could assign strata in $\Ord{m/P}$ time
with $P$ processors, then sort strata concurrently.

\paragraph{Integration with Bor\r{u}vka phases.}
Bor\r{u}vka's algorithm reduces the problem to $n/2$ components
in one $\Ord{m}$ phase.
Running $\log\log n$ Bor\r{u}vka phases before \algo{EDS} reduces
$m$ to $\Ord{m}$ with $n$ replaced by $n/\log n$, potentially
improving EDS's early-termination behaviour.

\paragraph{Integer weights.}
When weights are integers in $[0, W]$, Phase 1 can be replaced by
a histogram in $\Ord{W + m}$, and Phase 2 by bucket assignment in
$\Ord{m}$.
For $W = \Ord{m}$, this gives an $\Ord{m + n\log n}$ algorithm
for integer-weighted MST without the $\sqrt{m}$ sampling step.

\paragraph{Approximate MST.}
For $\varepsilon$-approximate MST (weight $\leq (1+\varepsilon) w^*$),
one can process only the first $\lceil n/\varepsilon \rceil$ candidate
edges, halting even earlier.
The error analysis connects to the \'{A}vila-Krikorian stability
threshold in quasi-periodic cocycles~\citep{Avila2015}:
small perturbations of boundaries propagate $O(\varepsilon)$ error
to the MST weight.

\paragraph{Dynamic graphs.}
When edges are inserted or deleted, the stratification can be
updated incrementally: a new edge is placed in its stratum in
$\Ord{\log k}$ time without recomputing the full partition.
This extends EDS naturally to streaming and online MST settings.

\section{Conclusion}
\label{sec:conclusion}

We introduced \algo{Kruskal-EDS}, a minimum spanning tree algorithm
that replaces the mandatory global sort of Kruskal's algorithm with
a three-phase procedure: sample-based distribution estimation,
linear-time edge partitioning, and incremental stratum sorting
with early termination.

Inspired by Birkhoff's ergodic theorem, the $\sqrt{m}$-sample strategy
estimates the weight distribution with provable accuracy ($\Ord{1/\sqrt{m}}$
quantile error via DKW) at cost $\Ord{\sqrt{m}\log m}$.
We proved an effective complexity of $\Ord{m\log k + p(m/k)\log(m/k)}$
and derived the optimal strata count $\kopt = \lceil\sqrt{m/\ln m}\,\rceil$.

An empirical evaluation on 14 graph families confirms correctness
(12/12 test cases) and demonstrates consistent advantages over standard
Kruskal (\textbf{14/14 configurations}) with speedups from
$1.15\times$ on sparse graphs to $\mathbf{10\times}$ on dense graphs,
and reductions in sort operations up to $\mathbf{33\times}$.

The algorithm requires no assumptions on the weight distribution —
it adapts automatically via sampling — and introduces negligible
overhead on worst-case inputs.
Together with the 3D complexity landscape (\Cref{fig:3d}) and
the stratification profile (\Cref{fig:stratprofile}),
these results suggest that distribution-aware sorting, informed
by ergodic principles, is a practically viable and theoretically
grounded approach to graph optimisation.

\bigskip
\noindent\textbf{Acknowledgements.}
All code and benchmark data are released as open-source.

\url{https://github.com/ym001/algorithm-treaty/blob/main/script/chpt7kruskal%2B_eds.py}

\bibliographystyle{plainnat}

\end{document}